\documentclass{llncs}
\usepackage{makeidx}  
\usepackage[ruled,linesnumbered,commentsnumbered]{algorithm2e}
\usepackage{pgf,tikz}
\usetikzlibrary{arrows}
\usetikzlibrary{decorations.markings}

\usepackage{color,soul}
\usepackage{amssymb}
\usepackage{graphicx}
\usepackage{subfigure}
\usepackage{color}
\usepackage{soul}
\tikzset{->-/.style={decoration={
  markings,
  mark=at position .97 with {\arrow{>}}},postaction={decorate}}}

\usepackage{pgf,tikz}
\usetikzlibrary{positioning}
\usetikzlibrary{arrows}

\begin{document}

\newcommand{\listbubbles}{\ensuremath{\mathtt{enumerate\_bubbles}}}

\title{A polynomial delay algorithm for the enumeration of bubbles with length constraints in directed graphs and its application to the detection of alternative splicing in RNA-seq data}
%
\titlerunning{Enumerating bubbles}  
%

\author{Gustavo Sacomoto \and Vincent Lacroix \and Marie-France Sagot}
%
\tocauthor{}
\institute{
  INRIA Rh\^one-Alpes, 38330 Montbonnot Saint-Martin, France 
\and
  Universit\'e de Lyon, F-69000 Lyon; Universit\'e Lyon 1; CNRS, UMR5558,
  Laboratoire de Biom\'etrie et Biologie \'Evolutive, F-69622 Villeurbanne, France \\
  \email{\{gustavo.sacomoto, Marie-France.Sagot\}@inria.fr}, \email{vincent.lacroix@univ-lyon1.fr}  
}

\maketitle    

\begin{abstract}
We present a new algorithm for enumerating bubbles with length
constraints in directed graphs. This problem arises in
transcriptomics, where the question is to identify all alternative
splicing events present in a sample of mRNAs sequenced by RNA-seq.
This is the first polynomial-delay algorithm for this problem and we
show that in practice, it is faster than previous approaches.  This
enables us to deal with larger instances and therefore to discover
novel alternative splicing events, especially long ones, that were
previously overseen using existing methods.
\end{abstract}


\section{Introduction}
Transcriptomes of model or non model species can now be studied by
sequencing, through the use of RNA-seq, a protocol which enables us to
obtain, from a sample of RNA transcripts, a (large) collection of
(short) sequencing reads, using Next Generation Sequencing (NGS)
technologies \cite{Burge,Mortazavi}. Nowadays, a typical experiment
produces 100M reads of 100nt each. However, the original RNA molecules
are longer (typically 500-3000nt) and the general computational
problem in the area is then to be able to assemble the reads in order
to reconstruct the original set of transcripts.  This problem is not
trivial for mainly two reasons. First, genomes contain repeats that
may be longer than the read length. Hence, a read does not necessarily
enable to identify unambiguously the locus from which the transcript
was produced. Second, each genomic locus may generate several types of
transcripts, either because of genomic variants (i.e. there may exist
several alleles for a locus) or because of transcriptomic variants
(i.e. alternative splicing or alternative transcription start/end may
generate several transcripts from a single locus that differ by the
inclusion or exclusion of subsequences). Hence, if a read matches a
subsequence shared by several alternative transcripts, it is a priori
not possible to decide which of these transcripts generated the read.

General purpose transcriptome assemblers
\cite{Trinity,TransAbyss,Oases} aim at the general goal of identifying
all alternative transcripts, but because of the extensive use of
heuristics, they usually fail to identify infrequent transcripts, tend
to report several fragments for each gene, or fuse genes that share
repeats.  Local transcriptome assemblers \cite{Sacomoto2012}, on the
other hand, aim at a simpler goal, as they do not reconstruct full
length transcripts. Instead, they focus on reporting all variable
regions (polymorphisms): whether genomic (SNPs, indels) or
transcriptomic (alternative splicing events). They are much less
affected by the issue of repeats, since they focus only on the
variable regions. They can afford to be exact and therefore are able
to have access to infrequent transcripts.  The fundamental idea is
that each polymorphism corresponds to a recognizable pattern, called a
bubble in the de Bruijn graph built from the RNA-seq reads. In
practice, only bubbles with specific length constraints are of
interest. However, even with this restriction, the number of such
bubbles can be exponential in the size of the graph. Therefore, as
with other enumeration problems, the best possible algorithm is one
spending time polynomial in the input size between the output of two
bubbles, i.e. a polynomial delay algorithm

In this paper, we introduce the first polynomial delay algorithm to
enumerate all bubbles with length constraints in a weighted directed
graph.  Its complexity in the best theoretical case for general graphs
is $O(n(m+n \log n))$ (Section~\ref{sec:alg_delay}) where $n$ is the
number of vertices in the graph, $m$ the number of arcs.  In the
particular case of de Bruijn graphs, the complexity is $O(n(m+n \log
\alpha))$ (Section~\ref{subsec:dijkstra}) where $\alpha$ is a constant
related to the length of the skipped part in an alternative splicing
event. In practice, an algorithmic solution in $O(nm\log n)$
(Section~\ref{subsec:comp_kissplice}) appears to work better on de
Bruijn graphs built from such data.  We implemented the latter, show
that it is more efficient than previous approaches and outline that it
enables us to discover novel long alternative splicing events.

\section{De Bruijn graphs and alternative splicing} \label{sec:debruijn_and_as}
A \emph{de Bruijn graph} (DBG) is a directed graph $G=(V,A)$ whose
vertices $V$ are labeled by words of length $k$ over an alphabet
$\Sigma$.  An arc in $A$ links a vertex $u$ to a vertex $v$ if the
suffix of length $k-1$ of $u$ is equal to the prefix of $v$.  The out
and the in-degree of any vertex are therefore bounded by the size of
the alphabet $\Sigma$. In the case of NGS data, the $k$-mers
correspond to all words of length $k$ present in the reads of the
input dataset, and only those. In relation to the classical de Bruijn
graph for all possible words of size $k$, the DBG for NGS data may
then not be complete.  Given two vertices $s$ and $t$ in $G$, an
$(s,t)$-path is a path from $s$ to $t$. As defined in \cite{ourSPIRE},
by an $(s,t)$-bubble, we mean two vertex-disjoint $(s,t)$-paths. This
definition is, of course, not restricted to de Bruijn graphs.

As was shown in \cite{Sacomoto2012}, polymorphisms (i.e. variable
parts) in a transcriptome (including alternative splicing (AS) events)
correspond to recognizable patterns in the DBG that are precisely the
$(s,t)$-bubbles. Intuitively, the variable parts correspond to
alternative paths and the common parts correspond to the beginning and
end points of those paths. More formally, any process generating
patterns $awb$ and $aw'b$ in the sequences, with $a,b,w,w' \in
\Sigma^*$, $|a| \geq k, |b|\geq k$ and $w$ and $w'$ not sharing any
$k$-mer, creates a $(s,t)$-bubble in the DBG. In the special case of
AS events excluding mutually exclusive exons, since $w'$ is empty, one
of the paths corresponds to the \emph{junction} of $ab$, i.e. to
$k$-mers that contain at least one letter of each sequence. Thus the
number of vertices of this path in the DBG is predictable: it is at
most\footnote{The size is \emph{exactly} $k-1$ if $w$ has no common
  prefix with $b$ and no common suffix with $a$.}  $k-1$. An example
is given in Fig. \ref{fig:ex_bubble}. In practice \cite{Sacomoto2012},
an upper bound $\alpha$ to the other path and a lower bound $\beta$ on
both paths is also imposed. In other words, an AS event corresponds to
a $(s,t)$-bubble with paths $p_1$ and $p_2$ such that $p_1$ has at
most $\alpha$ vertices, $p_2$ at most $k-1$ and both have at least
$\beta$ vertices.
\begin{figure}[htb]
\center
\resizebox{!}{1.7cm}{%
  \begin{tikzpicture}[->,>=stealth',shorten >=1pt,auto,node distance=1.5cm,
      thick,main node/.style={rectangle,draw,font=\sffamily\bfseries}]

    \node[main node] (1) {\color{red}ACT};
    \node[main node] (2) [right of=1] {{\color{red}CT}{\color{blue}G}};
    \node[main node] (3) [above right of=2, yshift=-0.2cm] {{\color{red}T}{\color{black}GG}};
    \node[main node] (4) [right of=3] {\color{black}GGA};
    \node[main node] (5) [right of=4] {{\color{black}GA}{\color{cyan}G}};
    \node[main node] (6) [right of=5] {{\color{black}A}{\color{cyan}GC}};
    \node[main node] (7) [below right of=6] {\color{cyan}GCG};
    \node[main node] (8) [right=2.4cm of 2, yshift=-0.80cm] {{\color{red}T}{\color{cyan}GC}};
    
    \path[every node/.style={font=\sffamily\small}]
      (1) edge node [left] {} (2)
      (2) edge node [left] {} (3)
          edge node [left] {} (8)
      (3) edge node [left] {} (4)
      (4) edge node [left] {} (5)
      (5) edge node [left] {} (6)
      (6) edge node [left] {} (7)
      (8) edge node [left] {} (7);
  \end{tikzpicture}
}
\label{fig:ex_bubble}
\caption{DBG with $k=3$ for the sequences: {\color{red}
    ACT}{\color{black}GGA}{\color{cyan}GCG} ($awb$) and {\color{red}
    ACT}{\color{cyan}GCG} ($ab$). The pattern in the sequence
  generates a $(s,t)$-bubble, from {\color{red}CT}{\color{blue}G} to
  {\color{cyan}GCG}. In this case, $b=$ {\color{cyan}GCG} and $w=$ GGA
  have their first letter {\color{blue}G} in common, so the path
  corresponding to the junction $ab$ has $k-1-1 = 1$ vertex.}
\end{figure}
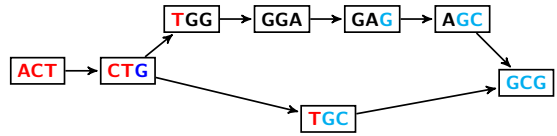

Given a directed graph $G$ with \emph{non-negative} arc weights $w: E
\mapsto \mathbb{Q}_{\geq 0}$, the length of the path $p = (v_0,v_1)
\ldots (v_{n-1}, v_n)$ is the sum of the weights of the edges in $p$
and is denoted by $|p|$. The distance, length of the shortest path,
from $u$ to $v$ is denoted by $d(u,v)$. We extend the definition of
bubble given above.

\begin{definition}[$(s,t,\alpha_1,\alpha_2)$-bubble] 
  A \emph{$(s,t, \alpha_1,\alpha_2)$-bubble} in a weighted directed
  graph is a $(s,t)$-bubble with paths $p_1,p_2$ satisfying $|p_1|
  \leq \alpha_1$ and $|p_2| \leq \alpha_2$.
\end{definition}

In practice, when dealing with DBGs built from NGS data, in a lossless
preprocessing step, all maximal non-branching linear paths of the
graph (i.e. paths containing only vertices with in and out-degree 1)
are compressed each into one single vertex, whose label corresponds to
the label of the path (i.e. it is the concatenation of the labels of
the vertices in the path without the overlapping part(s)). The
resulting graph is the \emph{compressed de Bruijn graph} (cDBG). In
the cDBG, the vertices can have labels larger than $k$, but an arc
still indicates a suffix-prefix overlap of size $k-1$. Finally, since
the only property of a bubble corresponding to an AS event is the
constraint on the length of the path, we can disregard the labels from
the cDBG and only keep for each vertex its label
length\footnote{Resulting in a graph with weights in the vertices. Here,
  however, we consider the weights in the arcs. Since this is more
  standard and, in our case, both alternatives are equivalent, we can
  transform one into another by splitting vertices or arcs.}. In this
way, searching for bubbles corresponding to AS events in a cDBG can be
seen as a particular case of looking for $(s,t, \alpha_1,
\alpha_2)$-bubbles satisfying the lower bound $\beta$ in a
non-negative weighted directed graph.

Actually, it is not hard to see that the enumeration of $(s,t,
\alpha_1,\alpha_2)$-bubbles, for all $s$ and $t$, satisfying the lower
bound $\beta$ is NP-hard. Indeed, deciding the existence of at least
one $(s,t, \alpha_1,\alpha_2)$-bubble, for some $s$ and $t$, with the
lower bound $\beta$ in a weighted directed graph where all the weights
are 1 is NP-complete. It follows by a simple reduction from the
Hamiltonian $st$-path problem~\cite{Cormen}: given a directed graph
$G = (V,E)$ and two vertices $s$ and $t$, build the graph $G'$ by
adding to $G$ the vertices $s'$ and $t'$, the arcs $(s,s')$ and
$(t,t')$, and a new path from $s'$ to $t'$ with exactly $|V|$
nodes. There is a $(x,y,|V|+2,|V|+2)$-bubble, for some $x$ and $y$,
satisfying the lower bound $\beta = |V| + 2$ in $G'$ if and only if
there is a Hamiltonian path from $s$ to $t$ in $G$.

From now on, we consider the enumeration of all
$(s,t,\alpha_1,\alpha_2)$-bubbles (without the lower bound) for a
given source (fixed $s$) in a non-negative weighted directed graph $G$
(not restricted to a cDBG). The number of vertices and arcs of $G$ is
denoted by $n$ and $m$, respectively.

\section{An $O(n (m + n \log n))$ delay algorithm} \label{sec:alg_delay}
In this section, we present an $O(n (m + n \log n))$ delay algorithm
to enumerate, for a fixed source $s$, all
$(s,t,\alpha_1,\alpha_2)$-bubbles in a general directed graph $G$ with
non-negative weights.  In a \emph{polynomial delay} enumeration
algorithm, the time elapsed between the output of two solutions is
polynomial in the instance size.  The pseudocode is shown in
Algorithm~\ref{listbubbles2}. It is important to stress that this
pseudocode uses high-level primitives, e.g. the tests in
lines~\ref{alg2:initial_test}, \ref{alg2:include} and
\ref{alg2:exclude}. An efficient implementation for the test in
line~\ref{alg2:include}, along with its correctness and analysis, is
implicitly given in Lemma~\ref{lem:test2}. This is a central result in
this section. For its proof we need Lemma~\ref{lem:dist}.

Algorithm~\ref{listbubbles2} uses a recursive strategy, inspired by
the binary partition method, that successively divides the solution
space at every call until the considered subspace is a singleton.  In
order to have a more symmetric structure for the subproblems, we
define the notion of a \emph{pair of compatible paths}, which is an
object that generalizes the definition of a
$(s,t,\alpha_1,\alpha_2)$-bubble. Given two vertices $s_1,s_2 \in V$
and upper bounds $\alpha_1, \alpha_2 \in \mathbb{Q}_{\geq0}$, the
paths $p_1 = s_1 \leadsto t_1$ and $p_2 = s_2 \leadsto t_2$ are a
\emph{pair of compatible paths} for $s_1$ and $s_2$ if $t_1 = t_2$,
$|p_1| \leq \alpha_1$, $|p_2| \leq \alpha_2$ and the paths are
internally vertex-disjoint. Clearly, every
$(s,t,\alpha_1,\alpha_2)$-bubble is also a pair of compatible paths
for $s_1 = s_2 = s$ and some $t$.

Given a vertex $v$, the set of out-neighbors of $v$ is denoted by
$\delta^+(v)$. Let now $\mathcal{P}_{\alpha_1,\alpha_2}(s_1,s_2,G)$ be
the set of all pairs of compatible paths for $s_1$, $s_2$, $\alpha_1$
and $\alpha_2$ in $G$. We have\footnote{The same relation is true
  using $s_1$ instead of $s_2$.} that:
\begin{equation} \label{eq:partition}
\mathcal{P}_{\alpha_1, \alpha_2}(s_1,s_2,G) = \mathcal{P}_{\alpha_1, \alpha_2}(s_1,s_2,G') 
                     \bigcup_{v \in \delta^+(s_2)} (s_2,v) \mathcal{P}_{\alpha_1, \alpha_2'} (s_1,v,G - s_2), 
\end{equation}
where $\alpha_2' = \alpha_2 - w(s_2,v)$ and $G' = G - \{(s_2,v) | v
\in \delta^+(s_2) \}$. In other words, the set of pairs of compatible
paths for $s_1$ and $s_2$ can be partitioned into:
$\mathcal{P}_{\alpha_1, \alpha_2'} (s_1,v,G - s_2)$, the sets of pairs
of paths containing the arc $(s_2,v)$, for each $v \in \delta^+(s_2)$;
and $\mathcal{P}_{\alpha_1, \alpha_2}(s_1,s_2,G')$, the set of pairs
of paths that do not contain any of them. Algorithm~\ref{listbubbles2}
implements this recursive partition strategy. The solutions are only
output in the leaves of the recursion tree (line~\ref{alg2:output}),
where the partition is always a singleton.  Moreover, in order to
guarantee that every leaf in the recursion tree outputs at least one
solution, we have to test if $\mathcal{P}_{\alpha_1, \alpha_2'}
(s_1,v,G - s_2)$ (and $\mathcal{P}_{\alpha_1, \alpha_2}(s_1,s_2,G')$)
is not empty before making the recursive call
(lines~\ref{alg2:include} and \ref{alg2:exclude}).

\begin{algorithm} 
\caption{$\listbubbles(s_1,\alpha_1,s_2,\alpha_2, B, G)$} \label{listbubbles2}
\If{$s_1 = s_2$}{ 
  \uIf{$B \neq \emptyset$}{
    output(B) \\ \label{alg2:output}
    \bf return 
  } 
  \ElseIf{there is no $(s, t, \alpha_1,\alpha_2)$-bubble, where $s = s_1 = s_2$}{ \label{alg2:initial_test} 
    \bf return
  } 
} 
choose $u \in \{s_1,s_2\}$, such that $\delta^+(u) \neq \emptyset$ \\
\For{$v \in \delta^+(u)$}{
  \If{there is a pair of compatible paths using $(u,v)$ in $G$}{ \label{alg2:include}
    \uIf{$u = s_1$}{ 
      $\listbubbles(v, \alpha_1 - w(s_1,v), s_2, \alpha_2, B \cup (s_1,v), G - s_1)$ \label{alg2:rec_include1}
    } 
    \Else{ 
      $\listbubbles(s_1, \alpha_1, v, \alpha_2 - w(s_2,v), B \cup (s_2,v), G - s_2)$ \label{alg2:rec_include2}
    } 
  } 
}
\If{there is a pair of compatible paths in $G - \{(u,v) | v \in \delta^+(u) \}$}{ \label{alg2:exclude}
  $\listbubbles(v, \alpha_1, s_2, \alpha_2, B, G - \{(u,v) | v \in \delta^+(u) \})$ \label{alg2:rec_exclude}
}
\end{algorithm}

The correctness of Algorithm~\ref{listbubbles2} follows directly from
the relation given in Eq.~\ref{eq:partition} and the correctness of
the tests performed in lines \ref{alg2:include} and \ref{alg2:exclude}. In the
remaining of this section, we describe a possible implementation for the
tests, prove correctness and analyze the time complexity. Finally, we
prove that Algorithm~\ref{listbubbles2} has an $O(n(m + n \log n))$
delay.

\begin{lemma} \label{lem:dist}
  There exists a pair of compatible paths for $s_1 \neq s_2$ in $G$ if
  and only if there exists $t$ such that $d(s_1,t) \leq \alpha_1$ and
  $d(s_2, t) \leq \alpha_2$.
\end{lemma}
\begin{proof}
  Clearly this is a necessary condition. Let us prove that it is also
  sufficient. Consider the paths $p_1 = s_1 \leadsto t$ and $p_2 = s_2
  \leadsto t$, such that $|p_1| \leq \alpha_1$ and $|p_2| \leq
  \alpha_2$. Let $t'$ be the first vertex in common between $p_1$ and
  $p_2$. The sub-paths $p_1' = s_1 \leadsto t'$ and $p_2' = s_2
  \leadsto t'$ are internally vertex-disjoint, and since the weights
  are non-negative, they also satisfy $|p_1'| \leq |p_1| \leq \alpha_1$
  and $|p_2'| \leq |p_2| \leq \alpha_2$.
\end{proof}

Using this lemma, we can test for the existence of a pair of
compatible paths for $s_1 \neq s_2$ in $O(m + n \log n)$ time. Indeed,
let $T_1$ be a shortest path tree of $G$ rooted in $s_1$ and truncated
at distance $\alpha_1$, the same for $T_2$, meaning that, for any
vertex $w$ in $T_1$ (resp. $T_2$), the tree path between $s_1$ and $w$
(resp. $s_2$ and $w$) is a shortest one.  It is not difficult to prove
that the intersection $T_1 \cap T_2$ is not empty if and only if there
is a pair of compatible paths for $s_1$ and $s_2$ in $G$. Moreover,
each shortest path tree can be computed in $O(m + n\log n)$ time,
using Dijkstra's algorithm~\cite{Cormen}. Thus, in order to test for
the existence of a $(s, t, \alpha_1,\alpha_2)$-bubble for some $t$ in
$G$, we can test, for each arc $(s,v)$ outgoing from $s$, the
existence of a pair of compatible paths for $s \neq v$ and $v$ in
$G$. Since $s$ has at most $n$ out-neighbors, we obtain
Lemma~\ref{lem:initial_test}.
  
\begin{lemma} \label{lem:initial_test}
  The test of line \ref{alg2:initial_test} can be performed in $O(n (m
  + n \log n))$.
\end{lemma}

The test of line~\ref{alg2:include} could be implemented using the
same idea. For each $v \in \delta^+(u)$, we test for the existence of
a pair of compatible paths for, say, $u = s_2$ (the same would apply
for $s_1$) and $v$ in $G - u$, that is $v$ is in the subgraph of $G$
obtained by eliminating from $G$ the vertex $u$ and all the arcs
incoming to or outgoing from $u$.  This would lead to a total cost of
$O(n(m+ n \log n))$ for all tests of line~\ref{alg2:include} in each
call. However, this is not enough to achieve an $O(n(m + n \log n))$
delay. In Lemma~\ref{lem:test2}, we present an improved strategy to
perform these tests in $O(m+ n \log n)$ total time.

\begin{lemma} \label{lem:test2}
  The test of line~\ref{alg2:include}, for all $v \in \delta^+(u)$,
  can be performed in $O(m + n \log n)$ total time.
\end{lemma}
\begin{proof}
  Let us assume that $u = s_2$, the case $u = s_1$ is
  symmetric. From Lemma~\ref{lem:dist}, for each $v \in \delta^+(u)$,
  we have that deciding if there exists a pair of compatible paths for
  $s_1$ and $s_2$ in $G$ that uses $(u,v)$ is equivalent to deciding
  if there exists $t$ satisfying (i) $d(s_1,t) \leq \alpha_1$ and (ii)
  $d(v,t) \leq \alpha_2 - w(u,v)$ in $G - u$.

  First, we compute a shortest path tree rooted in $s_1$ for
  $G-u$. Let $V_{\alpha_1}$ be the set of vertices at a distance at
  most $\alpha_1$ from $s_1$. We build a graph $G'$ by adding a new
  vertex $r$ to $G-u$, and for each $y \in V_{\alpha_1}$, we add the
  arcs $(y,r)$ with weight $w(y,r) = 0$.  We claim that there exists
  $t$ in $G-u$ satisfying conditions (i) and (ii) if and only if
  $d(v,r) \leq \alpha_2 - w(u,v)$ in $G'$. Indeed, if $t$ satisfies
  (i) we have that the arc $(t,r)$ is in $G'$, so $d(t,r) = 0$. From
  the triangle inequality and (ii), $d(v,r) \leq d(v,t) + d(t,r) =
  d(v,t) \leq \alpha_2 - w(u,v)$. The other direction is trivial.

  Finally, we compute a shortest path tree $T_r$ rooted in $r$ for the
  reverse graph $G'^R$, obtained by reversing the direction of the
  arcs of $G'$. With $T_r$, we have the distance from any vertex to
  $r$ in $G'$, i.e. we can answer the query $d(v,r) \leq \alpha_2 -
  w(u,v)$ in constant time. Observe that the construction of $T_r$
  depends only on $G-u$, $s_1$ and $\alpha_1$, i.e. $T_r$ is the same
  for all out-neighbors $v \in \delta^+(u)$. Therefore, we can build
  $T_r$ only once in $O(m + n \log n)$ time, with two iterations of
  Dijkstra's algorithm, and use it to answer each test of
  line~\ref{alg2:include} in constant time.
\end{proof}

\begin{theorem}
  Algorithm~\ref{listbubbles2} has $O(n(m + n \log n))$ delay.
\end{theorem}
\begin{proof}
  The height of the recursion tree is bounded by $2n$ since at each
  call the size of the graph is reduced either by one vertex
  (lines~\ref{alg2:rec_include1} and \ref{alg2:rec_include2}) or all
  its out-neighborhood (line~\ref{alg2:rec_exclude}). After at most
  $2n$ recursive calls, the graph is empty. Since every leaf of the
  recursion tree outputs a solution and the distance between two
  leaves is bounded by $4n$, the delay is $O(n)$ multiplied by the
  cost per node (call) in the recursion tree. From
  Lemma~\ref{lem:dist}, line~\ref{alg2:exclude} takes $O(m + n \log n)$
  time, and from Lemma~\ref{lem:test2}, line~\ref{alg2:include} takes
  $O(m + n \log n)$ total time. This leads to an $O(m + n \log n)$ time per call,
  excluding line~\ref{alg2:initial_test}. Lemma~\ref{lem:initial_test}
  states that the cost for the test in line~\ref{alg2:initial_test} is
  $O(n(m + n \log n))$, but this line is executed only once, at the
  root of the recursion tree. Therefore, the delay is $O(n (m + n \log
  n))$.
\end{proof}

\section{Implementation and experimental results}
We now discuss the details necessary for an efficient implementation
of Algorithm~\ref{listbubbles2} and the results on two sets of
experimental tests. For the first set, our goal is to compare the
running time of Dijkstra's algorithm (for typical DBGs arising from
applications) using several priority queue implementations. With the
second set, our objective is to compare an implementation of
Algorithm~\ref{listbubbles2} to the {\sc Kissplice} algorithm
\cite{Sacomoto2012}.  For both cases, we retrieved from the
\emph{Short Read Archive} (accession code ERX141791) 14M Illumina 79bp
single-ended reads of a \emph{Drosophila melanogaster} RNA-seq
experiment. We then built the de Bruijn graph for this dataset with $k
= 31$ using the {\sc Minia} algorithm~\cite{Rayan2012}. In order to
remove likely sequencing errors, we discarded all $k$-mers that are
present less than 3 times in the dataset. The resulting graph
contained 22M $k$-mers, which after compressing all maximal linear
paths, corresponded to 600k vertices.

In order to perform a fair comparison with {\sc Kissplice}, we
pre-processed the graph as described in \cite{Sacomoto2012}. Namely,
we decomposed the underlying undirected graph into biconnected
components (BCCs) and compressed all non-branching bubbles with equal
path lengths. In the end, after discarding all BCCs with less than 4
vertices (as they cannot contain a bubble), we obtained 7113 BCCs, the
largest one containing 24977 vertices.  This pre-processing is
lossless, i.e. every bubble in the original graph is entirely
contained in exactly one BCC. In {\sc Kissplice}, the enumeration is
then done in each BCC independently.

\subsection{Dijkstra's algorithm with different priority queues} \label{subsec:dijkstra}
Dijkstra's algorithm is an important subroutine of
Algorithm~\ref{listbubbles2} that may have a big influence on its
running time. Actually, the time complexity of
Algorithm~\ref{listbubbles2} can be written as $O(n c(n,m))$, where
$c(n,m)$ is the complexity of Dijkstra's algorithm. There are several
variants of this algorithm~\cite{Cormen}, with different complexities
depending on the priority queue used, including binary heaps ($O(m
\log n)$) and Fibonacci heaps ($O(m + n \log n)$). In the particular
case where all the weights are non-negative integers bounded by $C$,
Dijkstra's algorithm can be implemented using radix heaps ($O(m + n
\log C)$) \cite{Tarjan}. As stated in
Section~\ref{sec:debruijn_and_as}, the weights of the de Bruijn graphs
considered here are integer, but not necessarily bounded.  However, we
can remove from the graph all arcs with weights greater than
$\alpha_1$ since these are not part of any $(s,t,\alpha_1,
\alpha_2)$-bubble. This results in a complexity of $O(m + n \log
\alpha_1)$ for Dijkstra's algorithm.

\begin{figure}[Htbp]
  \center
  \includegraphics[width=7cm]{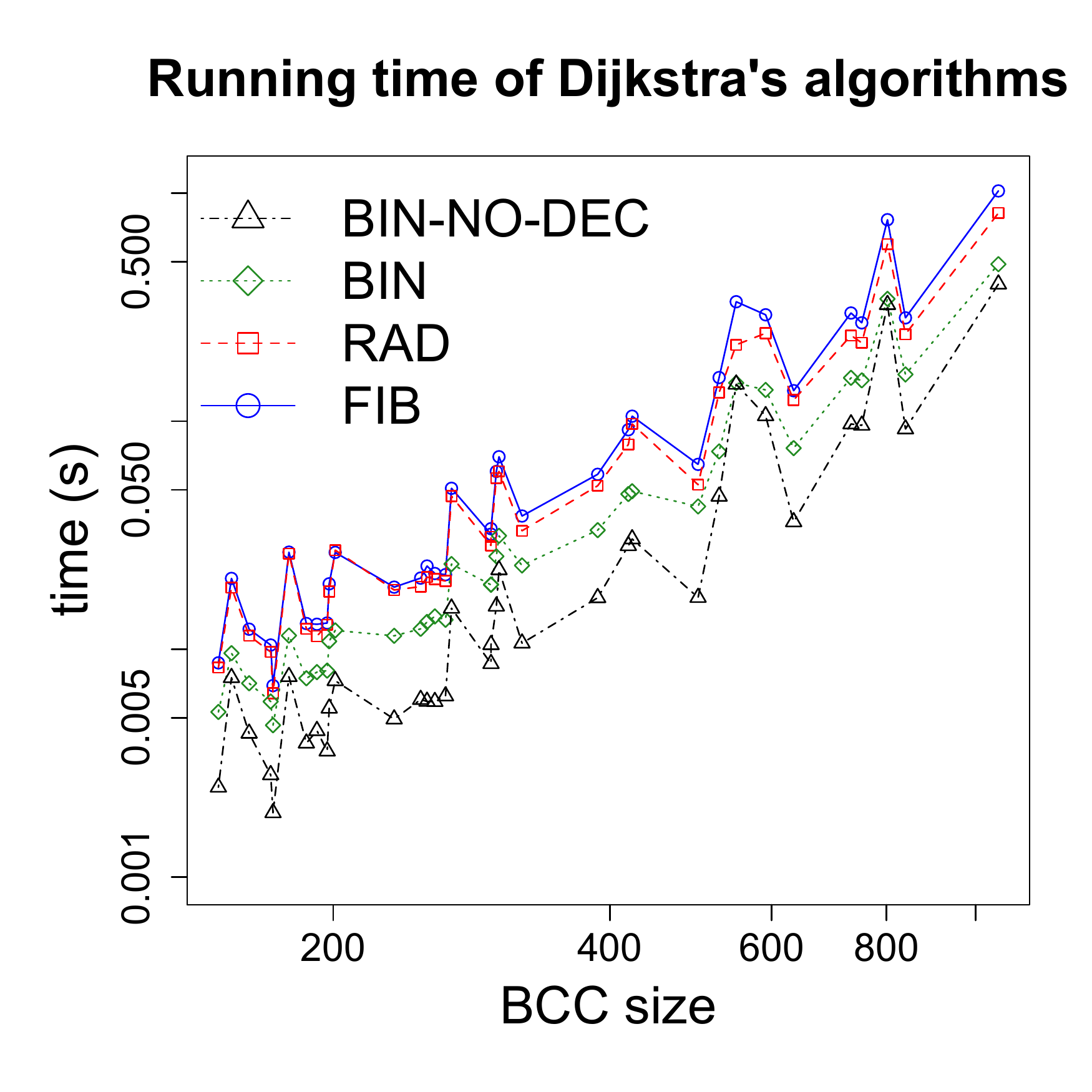}\label{fig:running_time_dijkstra}
  \caption{Running times for each version of Dijkstra's algorithm:
    using Fibonacci heaps (FIB), using radix heaps (RAD), using binary
    heaps (BIN) and using binary heaps without the decrease-key
    operation (BIN-NO-DEC).  The tests were done including all BCCs
    with more than 150 vertices. Both axes are in logarithmic
    scale.}  \label{fig:dijkstra}
\end{figure}

We implemented four versions of Lemma~\ref{lem:initial_test} (for
deciding whether there exists a $(s,t,\alpha_1, \alpha_2)$-bubble for
a given $s$) each using a different version of Dijkstra's algorithm:
with Fibonacci heaps (FIB), with radix heaps (RAD), with binary heaps
(BIN) and with binary heaps without decrease-key operation
(BIN-NO-DEC). The last version is Dijkstra's modified in order not to
use the decrease-key operation so that we can use a simpler binary
heap that does not support such operation~\cite{Chen}.  We then ran
the four versions, using $\alpha_1 = 1000$ and $\alpha_2 = 2k - 2 =
60$, for each vertex in all the BCCs with more than 150 vertices. The
results are shown\footnote{The results for the largest BCC were
  omitted from the plot to improve the visualization. It took 942.15s
  for FIB and 419.84s for BIN-NO-DEC.}  in
Fig.~\ref{fig:dijkstra}. Contrary to the theoretical predictions, the
versions with the best complexities, FIB and RAD, have the worst
results on this type of instances. It is clear that the best version
is BIN-NO-DEC, which is at least 2.2 times and at most 4.3 times
faster than FIB.  One of the factors possibly contribuiting to a
better performance of BIN and BIN-NO-DEC is the fact that cDBGs, as
stated in Section~\ref{sec:debruijn_and_as}, have bounded degree and
are therefore sparse.

\subsection{Comparison with the {\sc Kissplice} algorithm} \label{subsec:comp_kissplice}

In this section, we compare Algorithm~\ref{listbubbles2} to the {\sc
  Kissplice} (version 1.8.1) enumeration
algorithm~\cite{Sacomoto2012}. To this purpose, we implemented
Algorithm~\ref{listbubbles2} using Dijkstra's algorithm with binary
heaps without the decrease-key operation for all shortest paths
computation. In this way, the delay of Algorithm~\ref{listbubbles2}
becomes $O(nm \log n)$, which is worse than the one using Fibonacci or
radix heaps, but is faster in practice. The goal of the {\sc
  Kissplice} enumeration is to find all the potential alternative
splicing events in a BCC, i.e. to find all
$(s,t,\alpha_1,\alpha_2)$-bubbles satisfying also the lower bound
constraint (Section~\ref{sec:debruijn_and_as}).  In order to compare
{\sc Kissplice} to Algorithm~\ref{listbubbles2}, we (naively) modified
the latter so that, whenever a $(s,t,\alpha_1,\alpha_2)$-bubble is
found, we check whether it also satisfies the lower bound constraints
and output it only if it does.

In {\sc Kissplice}, the upper bound $\alpha_1$ is an open parameter,
$\alpha_2 = k-1$ and the lower bound is $k - 7$. Moreover, there are
two stop conditions: either when more than 10000
$(s,t,\alpha_1,\alpha_2)$-bubbles satisfying the lower bound
constraint have been enumerated or a 900s timeout has been reached.
We ran both {\sc Kissplice} (version 1.8.1) and the modified
Algorithm~\ref{listbubbles2}, with the stop conditions, for all 7113
BCCs, using $\alpha_2 = 60$, a lower bound of $54$ and $\alpha_1 =
250,500,750$ and $1000$. The running times for all BCCs with more than
150 vertices (there are 37) is shown\footnote{The BCCs where
  \emph{both} algorithms reach the timeout were omitted from the plots
  to improve the visualization. For $\alpha_1 = 250, 500, 750$ and
  $1000$ there are 1, 2, 3 and 3 BCCs omitted, respectively.}  in
Fig.~\ref{fig:running_time}. For the BCCs smaller than 150 vertices,
both algorithms have comparable (very small) running times. For
instance, with $\alpha_1 = 250$, {\sc Kissplice} runs in 17.44s for
\emph{all} 7113 BCCs with less than 150 vertices, while
Algorithm~\ref{listbubbles2} runs in 15.26s.

The plots in Fig.~\ref{fig:running_time} show a trend of increasing
running times for larger BCCs, but the graphs are not very smooth,
i.e. there are some sudden decreases and increases in the running
times observed. This is in part due to the fact that the time complexity of 
Algorithm~\ref{listbubbles2} is output sensitive. The delay of the
algorithm is $O(nm \log n)$, but the total time complexity is
$O(|\mathcal{B}|nm \log n)$, where $|\mathcal{B}|$ is the number of
$(s,t,\alpha_1,\alpha_2)$-bubbles in the graph. The number of bubbles
in the graph depends on its internal structure. A large graph does not
necessarily have a large number of bubbles, while a small graph may have
an exponential number of bubbles. Therefore, the value of
$|\mathcal{B}|nm \log n$ can decrease by increasing the size of the
graph.

Concerning now the comparison between the algorithms, as we can see in
Fig.~\ref{fig:running_time}, Algorithm~\ref{listbubbles2} is usually
several times faster (keep in mind that the axes are in logarithmic
scale) than {\sc Kissplice}, with larger differences when $\alpha_1$
increases (10 to 1000 times faster when $\alpha_1 = 1000$).  In some
instances however, {\sc Kissplice} is faster than
Algorithm~\ref{listbubbles2}, but (with only one exception for
$\alpha_1 = 250$ and $\alpha_1 = 500$) they correspond either to very
small instances or to cases where only 10000 bubbles were enumerated
and the stop condition was met.  Finally, using
Algorithm~\ref{listbubbles2}, the computation finished within 900s for
all but 3 BCCs, whereas using {\sc Kissplice}, 11 BCCs remained
unfinished after 900s.  The improvement in time therefore enables us to
have access to bubbles that could not be enumerated with the previous
approach.

\begin{figure}[Htbp]
  \center
  \subfigure[]{\includegraphics[width=5.9cm]{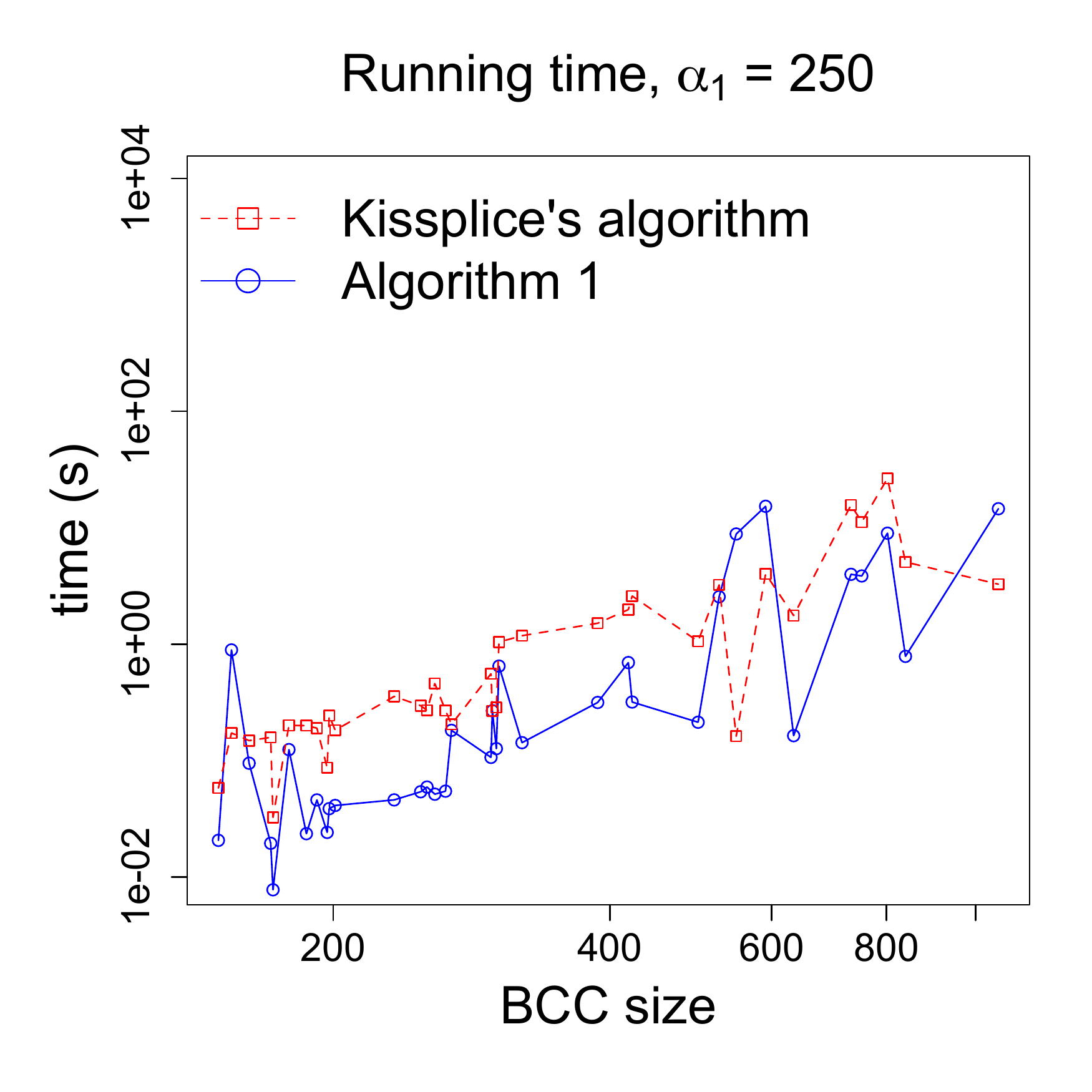}\label{fig:running_time_250}} 
  \subfigure[]{\includegraphics[width=5.9cm]{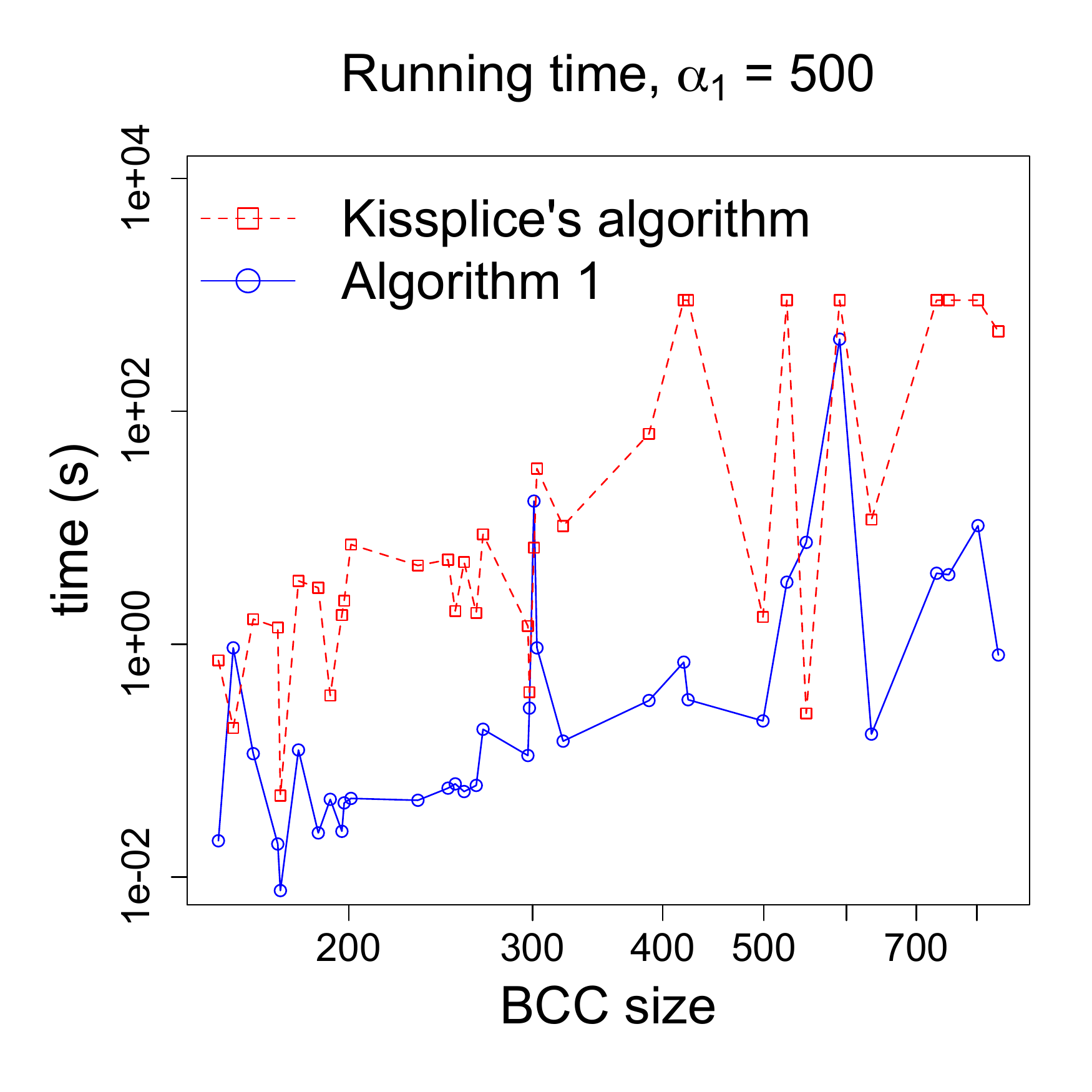}\label{fig:running_time_500}} \\
  \subfigure[]{\includegraphics[width=5.9cm]{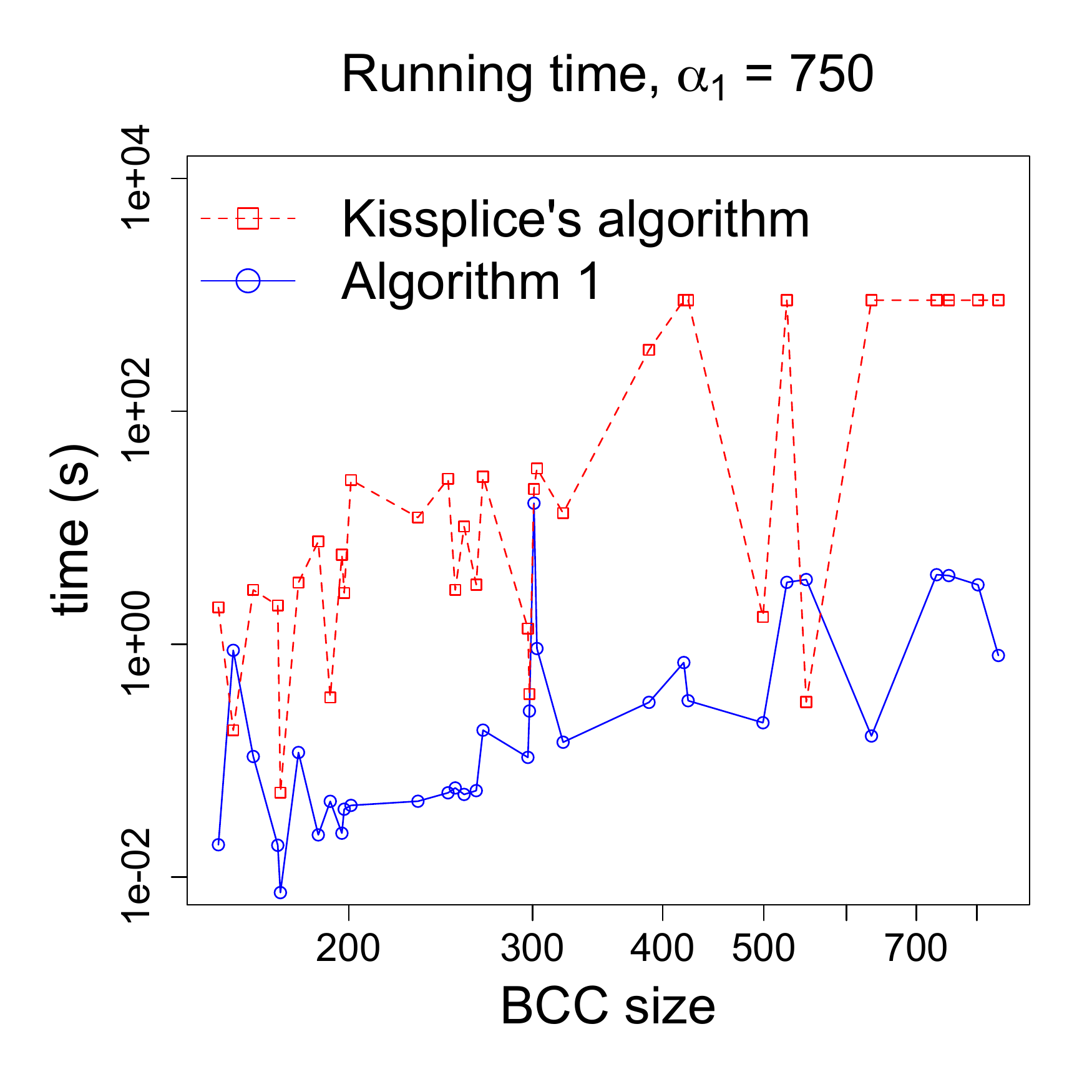} \label{fig:running_time_750}} 
  \subfigure[]{\includegraphics[width=5.9cm]{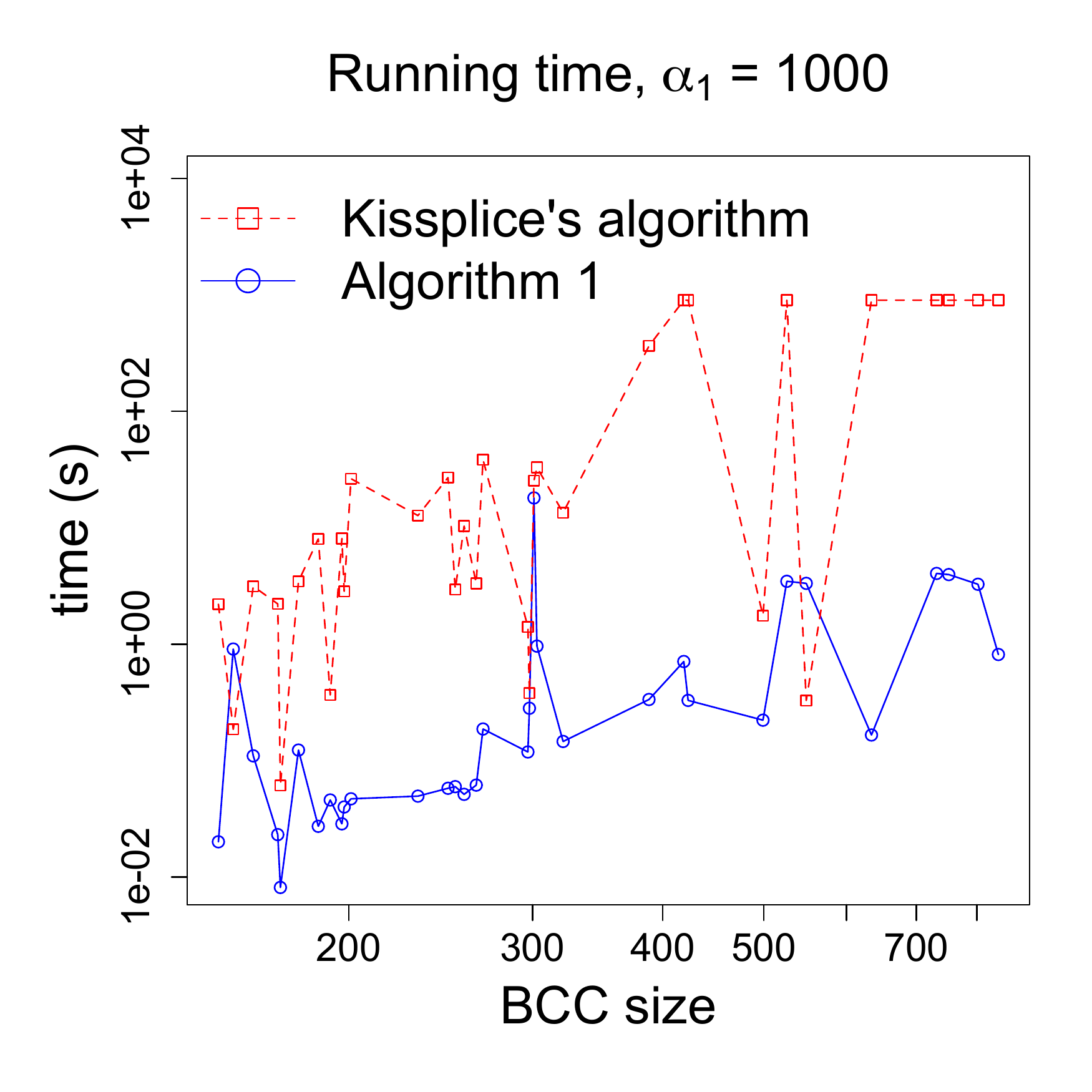}\label{fig:running_time_1000}}
  \caption{Running times of Algorithm~\ref{listbubbles2} and
    of the {\sc Kissplice} algorithm~\cite{Sacomoto2012} for all the BCCs with
    more than 150 vertices. Each graph (a), (b), (c) and (d) shows the
    running time of both algorithms for $\alpha_1 = 250, 500, 750$ and
    $1000$, respectively.} \label{fig:running_time}
\end{figure} 

\subsection{On the usefulness of larger values of $\alpha_1$}
In the implementation of {\sc Kissplice} \cite{KisspliceManual}, the
value of $\alpha_1$ was experimentally set to 1000 due to performance
issues, as indeed the algorithm quickly becomes impractical for larger
values. On the other hand, the results of
Section~\ref{subsec:comp_kissplice} suggest that
Algorithm~\ref{listbubbles2}, that is faster than {\sc Kissplice}, can
deal with larger values of $\alpha_1$.  From a biological point of
view, it is a priori possible to argue that $\alpha_1 = 1000$ is a
reasonable choice, because 87\% of annotated exons in Drosophila
indeed are shorter than 1000nt \cite{Refseq}. However, missing the top
13\% may have a big impact on downstream analyses of AS, not to
mention the possibility that not yet annotated AS events could be
enriched in long skipped exons.  In this section, we outline that
larger values of $\alpha_1$ indeed produces more results that are
biologically relevant. For this, we exploit another RNA-seq dataset,
with deeper coverage.

\begin{figure}[Htbp]
  \center
  \includegraphics[width=12cm]{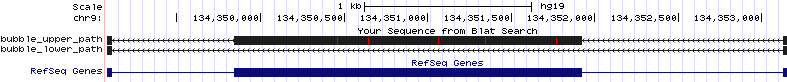}
  \label{fig:exon_skipping}
  \caption{One of the bubbles with longest path larger than 1000 bp
    found by Algorithm~\ref{listbubbles2} with the corresponding
    sequences mapped to the reference genome and visualized using the UCSC
    Genome Browser. The first two lines correspond to the sequences of, respectively,
    the shortest (exon exclusion variant) and longest paths of the
    bubble mapped to the genome.  The blue lines are
    the UCSC human transcript annotations. }
\end{figure}

To this purpose, we retrieved 32M RNA-seq reads from the human brain
and 39M from the human liver from the Short Read Archive (accession
number ERP000546). Next, we built the de Bruijn graph with $k=31$ for
both datasets, then merged and decomposed the DBG into 5692 BCCs
(containing more than 10 vertices). We ran
Algorithm~\ref{listbubbles2} for each BCC with $\alpha_1 = 5000$. It
took 4min25s for Algorithm~\ref{listbubbles2} to run on all BCCs,
whereas {\sc Kissplice}, even using $\alpha_1 = 1000$, took 31min45s,
almost 8 times more. There were 59 BCCs containing at least one bubble
with the length of the longest path strictly larger than 1000bp
potentially corresponding to alternative splicing events. In
Fig.~\ref{fig:exon_skipping}, we show one of those bubbles mapped to
the reference genome. It corresponds to an exon skipping in the PRRC2B
human gene, the skipped exon containing 2069 bp. While the transcript
containing the exon is annotated, the variant with the exon skipped is
not annotated.

Furthermore, we ran {\sc Trinity}~\cite{Trinity} (the most widely used
transcriptome assembler) on the same dataset and found that it was
unable to report this novel variant.  Our method therefore enables us
to find new AS events, reported by no other method. This is, of
course, just an indication of the usefulness of our approach when
compared to a full-transcriptome assembler. A more systematic
comparision with {\sc Trinity}, as done in \cite{Sacomoto2012}, is out
of the scope of this work.

\section{A natural generalization} \label{sec:gen}
For the sake of theoretical completeness, in this section, we extend
the definition of $(s,t,\alpha_1,\alpha_2)$-bubble to the case where
the length constraints concern $d$ vertex-disjoint paths, for an
arbitrary but fixed $d$.

\begin{definition}[$(s,t, A)$-$d$-bubble] \label{def:kbubble}
  Let $d$ be a natural number and $A = \{\alpha_1, \ldots, \alpha_d\}
  \subset \mathbb{Q}_{\geq0}$.  Given a directed weighted graph $G$
  and two vertices $s$ and $t$, an \emph{$(s,t,A)$-$d$-bubble} is a
  set of $d$ pairwise internally vertex-disjoint paths $\{p_1, \ldots
  p_d\}$, satisfying $p_i = s \leadsto t$ and $|p_i| \leq \alpha_i$,
  for all $i \in [1,d]$.
\end{definition} 

Analogously to $(s,t,\alpha_1,\alpha_2)$-bubbles, we can define two
variants of the enumeration problem: all bubbles with a given source
($s$ fixed) and all bubbles with a given source and target ($s$ and
$t$ fixed). In both cases, the first step is to decide the existence
of at least one $(s,t, A)$-$d$-bubble in the graph.

\begin{problem}[$(s,t,A)$-$d$-bubble decision problem] \label{prob:bounded_bubble_st}
  Given a non-negatively weighted directed graph $G$, two vertices
  $s,t$, a set $A = \{\alpha_1, \ldots, \alpha_d\} \subset
  \mathbb{Q}_{\geq0}$ and $d \in \mathbb{N}$, decide if there exists a
  $(s,t, A)$-$d$-bubble.
\end{problem} 

This problem is a generalization of the two-disjoint-paths problem
with a min-max objective function, which is
NP-complete~\cite{chung1990}. More formally, this problem can be
stated as follows: given a directed graph $G$ with non-negative
weights, two vertices $s,t \in V$, and a maximum length $M$, decide if
there exists a pair of vertex-disjoint paths such that the maximum of
their lengths is less than $M$. The $(s,t, A)$-$d$-bubble decision
problem, with $A = \{M,M\}$ and $d=2$, is precisely this problem.

\begin{problem}[$(s,*,A)$-$d$-bubble decision problem] \label{prob:bounded_bubble_s*}
  Given a non-negatively weighted directed graph $G$, a vertex $s$, a
  set $A = \{\alpha_1, \ldots, \alpha_d\} \subset \mathbb{Q}_{\geq0}$
  and $d \in \mathbb{N}$, decide if there exists a $(s,t,
  A)$-$d$-bubble, for some $t \in V$.
\end{problem} 

The two-disjoint-path problem with a min-max objective function is
NP-complete even for strictly positive weighted graphs. Let us reduce
Problem~\ref{prob:bounded_bubble_s*} to it.  Consider a graph $G$ with
strictly positive weights, two vertices $s,t \in V$, and a maximum
length $M$. Construct the graph $G'$ by adding an arc with weights $0$
from $s$ to $t$ and use this as input for the
$(s,*,\{M,M,0\})$-$3$-bubble decision problem. Since $G$ has strictly
positive weights, the only path with length $0$ from $s$ to $t$ in
$G'$ is the added arc. Thus, there is a $(s,*,\{M,M,0\})$-$3$-bubble
in $G'$ if and only if there are two vertex-disjoint paths in $G$ each
with a length $\leq M$.

Therefore, the decision problem for fixed $s$
(Problem~\ref{prob:bounded_bubble_st}) is NP-hard for $d \geq 2$, and
for fixed $s$ and $t$ (Problem~\ref{prob:bounded_bubble_s*}) is
NP-hard for $d \geq 3$.  In other words, the only tractable case is
the enumeration of $(s,t, A)$-$2$-bubbles with fixed $s$, the one
considered in Section~\ref{sec:alg_delay}.

\section{Conclusion}
We introduced a polynomial delay algorithm which enumerates all
bubbles with length constraints in directed graphs.  We show that it
is faster than previous approaches and therefore enables us to enumerate
more bubbles. These additional bubbles correspond to longer AS events,
overseen previously, but biologically very relevant.  As shown
in~\cite{Tarjan}, by combining radix and Fibonacci heaps in Dijkstra,
we can achieve a complexity in $O(n(m + n \sqrt{ \log \alpha_1)})$ for
Algorithm~\ref{listbubbles2} in cDGBs. The question whether this can
be improved, either by improving Dijkstra's algorithm (exploiting more
properties of a cDBG) or by using a different approach, remains open.

\paragraph{Acknowledgements} This work was funded by the ANR-12-BS02-0008 
(Colib'read); the French project ANR MIRI BLAN08-1335497; and the
European Research Council under the European Community's Seventh
Framework Programme (FP7 /2007-2013) / ERC grant agreement
no. [247073]10.

\bibliographystyle{plain} \bibliography{biblio} \nocite{*}

\end{document}